\documentclass[12pt,reqno]{amsart}
\usepackage{amssymb}
\usepackage[usenames,dvipsnames]{color}
\usepackage{euscript}
\usepackage{multicol}
\usepackage{amssymb}
\usepackage{amsmath}
\usepackage{times} 
\usepackage{colordvi}

\newtheorem{theorem}{Theorem}[section]
\newtheorem{corollary}[theorem]{Corollary}
\newtheorem{lemma}[theorem]{Lemma}

\numberwithin{equation}{section}

\begin{document}

\title[Some sharp spectral inequalities for Schr\"odinger operators on semi-axis]
{On some sharp spectral inequalities for Schr\"odinger operators on semi-axis}

\author{Pavel Exner}
\address{
Pavel Exner: Doppler Institute for Mathematical Physics and Applied Mathematics\\
B\v rehov\'{a} 7, 11519 Prague
\\ and Nuclear Physics Institute ASCR\\
25068 \v Re\v z near Prague, Czechia}
\email{exner@ujf.cas.cz}

\author{Ari Laptev}
\address{Ari Laptev: Imperial College London \\ 180 Queen's Gate \\ London SW7 2AZ \\ UK }
\email{a.laptev@imperial.ac.uk}

\author{Muhammad Usman}
\address{Muhammad Usman: Institut de Math\'ematiques de Bordeaux \\ Universit\'e Bordeaux 1 \\ 351, cours de la Lib\'eration \\ 33405 TALENCE, France}
\email{Muhammad.Usman@math.u-bordeaux1.fr}



\begin{abstract}
In this paper we obtain sharp Lieb-Thirring inequalities for a Schr\"odinger operator on semi-axis with a matrix potential and show how they can be used to other related problems. Among them are spectral inequalities on star graphs  and spectral inequalities for Schr\"odinger operators on half-spaces with Robin boundary conditions.
\end{abstract}

\maketitle

\section{Introduction}

\noindent
Let us consider a self-adjoint Schr\"odinger operator in $L^2({\Bbb R}^d)$
\begin{equation}\label{Schr}
H = -\Delta - V,
\end{equation}
where $V$ is a real-valued function. If the potential function $V$ decays rapidly enough, then  the spectrum of the operator $H$ typically is absolutely continuous on $[0,\infty)$. If $V$ has a non-trivial positive part, then $H$ might have finite or infinite number of  negative eigenvalues $\{-\lambda_n(H)\}$. If the number of negative eigenvalues is infinite, the point zero is the only possible accumulating point.
The inequalities
\begin{equation}\label{LTh}
\sum_n \lambda_n^\gamma
\leq \frac{R_{\gamma,d}}{(2\pi)^d}\, \iint_{\mathbb {R}^{2d}} (|\xi|^2 - V(x))_-^\gamma \, d\xi dx
\leq L_{\gamma,d}
\int_{{\Bbb R}^d}  V_-^{\gamma+\frac{d}{2}}\,dx\,
\end{equation}
are known as Lieb-Thirring bounds.
Here and in the following, $V_{\pm}=(|V|\pm V)/2$ denote
the positive and negative parts of the function $V$.

\medskip
\noindent
It is known that the inequality \eqref{LTh} holds true with some finite constants if and only if $\gamma\ge 1/2$, $d=1$;
$\gamma>0$, $d=2$ and $\gamma\ge0$, $d\ge3$.
There are examples showing that  \eqref{LTh} fails for $0\le\gamma<1/2$, $d=1$ and $\gamma=0$, $d=2$.

\medskip
\noindent
Almost all the cases except for $\gamma=1/2$, $d=1$ and $\gamma=0$, $d\ge3$ were justified in the original paper
of E.H.Lieb and W.Thirring  \cite{LT}. The critical case
$\gamma=0$, $d\ge3$ is known as the Cwikel-Lieb-Rozenblum inequality, see  \cite{Cw,L,Roz}. It was also proved in  \cite{Fe, LY,Con} and very recently by R. Frank \cite{Fr} using Rumin's approach. The remaining case $\gamma=1/2$, $d=1$ was verified by T.Weidl in \cite{W1}.

\medskip
\noindent
The sharp value of the constants $R_{\gamma,d}=1$ in \eqref{LTh} are known for the case $\gamma\ge3/2$ in all dimensions and it was first proved in  \cite{LT} and \cite{AizL} for $d=1$ and later in \cite{LW1, LW2} for any dimension. In this case
$$
L_{\gamma,d} = L_{\gamma,d}^{cl} := (2\pi)^{-d} \, \int_{\Bbb R^d} (1-|\xi|^2)_+^\gamma\, d\xi.
$$
The only other case where the sharp value of the constant $R_{\gamma,d}$ is known is the case $R_{1/2,1} = 2$.

\bigskip
\noindent
In this paper we consider a one-dimensional systems of Schr\"odinger operators acting in $L^2(\Bbb R_+,\mathbb{C}^N)$, $\Bbb R_+ = (0,\infty)$,  defined by
\begin{equation}\label{eq:1}
\mathcal H \varphi (x) = \Big(-\frac{d^2}{dx^2}\otimes \Bbb I - V(x)\Big) \, \varphi(x),
\qquad \varphi'(0) - {\mathfrak S}\, \varphi(0)=0,
\end{equation}
where $\Bbb I$ is the $N\times N$ identity matrix, $V$ is a Hermitian $N\times N$ matrix-function and ${\mathfrak S}$ is a $N\times N$  Hermitian matrix.

\noindent
Assuming that the potential $V$ generates only a discrete negative spectrum, we denote by $\{-\lambda_n\}$ the  negative eigenvalues of $\mathcal H$.

\medskip
\noindent
One of the main results of this paper is the following

%%%%%%%%%%%%%%%
%
\begin{theorem}\label{main}
Let ${\rm Tr} \, V^2 \in L^1({\Bbb R_+})$, $V\ge0$. 
Then the negative spectrum of the operator $\mathcal H$ defined in \eqref{eq:1} is discrete and the following Lieb-Thirring inequality for its eigenvalues $\{-\lambda_n\}$ holds
\begin{multline}\label{LTh_main}
\frac{3}{4}\lambda_1\, {\rm Tr}\, {\mathfrak S} +   \frac12\, (2\varkappa_1 -N) \, \lambda_1^{3/2}  + \sum_{n=2}^\infty \varkappa_n \lambda_n^{3/2} \\
\leq \frac{3}{16}\int_0^\infty {\rm Tr}\, V^2(x) \, dx+\frac{1}{4} \, {\rm Tr}\, {\mathfrak S^3},
\end{multline}
where $\varkappa_n$ is the multiplicity of the eigenvalue $-\lambda_n$.
\end{theorem}

%%%%%%%%%%%%%%%
\bigskip
\noindent
{\bf Examples.}

\noindent
{\bf 1.} {\it Let $V\equiv 0$ and $N=1$. Then the boundary value problem
$$
-\varphi''(x) = -\lambda \varphi(x) , \qquad \varphi'(0) - \sigma \varphi(0) = 0, \quad \sigma<0,
$$
has only one $L^2$-solution
$$
\varphi(x) = C \, e^{-\sqrt\lambda\, x}, \qquad  -\sqrt\lambda = \sigma.
$$
In this case the inequality \eqref{LTh_main} becomes saturated, $\frac34\, \sigma^3 - \frac12\, \sigma^3 \le \frac14\, \sigma^3$.}

\medskip
\noindent
{\bf 2.} {\it Let $N=2$, $V\equiv0$ and
$$
{\mathfrak S}=
\begin{pmatrix} \sigma & 0 \\ 0& -\alpha\sigma
\end{pmatrix}, \qquad \sigma<0
$$
{\bf 2a)} If $\alpha \ge0$ then the boundary value problem \eqref{eq:1} has one negative eigenvalue  $-\lambda$ of multiplicity one satisfying the identity  $-\sqrt\lambda =\sigma$. In this case $2\varkappa_1 -N = 0$ and the inequality \eqref{LTh_main} 
becomes
$$
3\lambda\, {\rm Tr}\, \mathfrak S = 3\lambda\, \sigma\, (1 - \alpha) \le (1-\alpha^3)\, \sigma^3 = {\rm Tr}\, {\mathfrak S^3},
$$
or
$$
3 (\alpha-1) \le \alpha^3-1. 
$$
which holds true for any $\alpha \ge0$.

\noindent
{\bf 2b)} If $-1 <\alpha<0$, then  the problem \eqref{eq:1} has two eigenvalues satisfying $-\sqrt{\lambda_1} = \sigma$ and $-\sqrt{\lambda_2} = -\alpha \sigma$ and  \eqref{LTh_main} is reduced to 
$$ 
3(\alpha -1) -4\alpha^3 \le \alpha^3 -1.
$$
{\bf 2c)} Finally, if  $\alpha = -1$, then $-\sqrt{\lambda_1} = \sigma$ is of multiplicity $\varkappa_1 = 2$ and \eqref{LTh_main} becomes 
identity.
}

\medskip
\noindent
Note that if ${\rm Tr}\, {\mathfrak S^3} \le 0$, then the inequality \eqref{LTh_main} implies
\begin{equation}\label{S<0}
\frac{3}{4}\lambda_1\, {\rm Tr}\, {\mathfrak S} +   \frac12\, (2\varkappa_1 -N) \, \lambda_1^{3/2}  + \sum_{n=2}^\infty \varkappa_n \lambda_n^{3/2} \leq \frac{3}{16}\int_0^\infty {\rm Tr}\, V^2(x) \, dx.
\end{equation}

\medskip
\noindent
The latter allows us to use the standard Aizenman-Lieb arguments \cite{AizL} and derive

%%%%%%
\begin{corollary}\label{AzL}
Let ${\rm Tr}\, {\mathfrak S^3} \le 0$, $V\ge0$ and ${\rm Tr}\,V^{\gamma + 1/2}(x)\in L^1(0,\infty)$. Then for any $\gamma\ge 3/2$ we have
\begin{multline*}
\frac{\mathcal B(\gamma-3/2,2)}{\mathcal B(\gamma-3/2,5/2)} \,  \frac{3}{4}\, \lambda_1^{\gamma-1/2} \, {\rm Tr}\, {\mathfrak S}  +  \frac12\, (2\varkappa_1 -N) \, \lambda_1^{\gamma} + \sum_{n=2}^\infty \varkappa_n \lambda_n^{\gamma}\\
 \leq
 L_{\gamma,1}^{cl} \int_{0}^\infty {\rm Tr}\,(V(x))^{\gamma + 1/2}\, dx,
\end{multline*}
where by $\mathcal B(p,q)$ we denote the classical Beta function
$$
\mathcal B(p,q) = \int_0^1 (1-t)^{q-1} t^{p-1}\, dt.
$$
\end{corollary}

\medskip
\begin{corollary}
If ${\mathfrak S} = 0$, then \eqref{eq:1} can be identified with the Neumann boundary value problem and
we obtain
$$
\frac12\, (2\varkappa_1 -N) \, \varkappa_1\lambda_1^{\gamma} + \sum_{n=2}^\infty \varkappa_n \lambda_n^{\gamma} \leq L_{\gamma,1}^{cl} \, \int_0^\infty {\rm Tr}\, (V(x))^{\gamma+1/2} \, dx, \quad \gamma\ge3/2.
$$
\end{corollary}
%%%%%%%%%%%%%%%%%%%%

\medskip
\noindent
{\bf Remark.}

\noindent
{\it Note that in the scalar case $N=1$ we obtain
\begin{equation}\label{N-scalar}
\frac12\, \lambda _1^{\gamma} + \sum_{n=2}^\infty \lambda_n^{\gamma} \leq L_{\gamma,1}^{cl} \, \int_0^\infty V^{\gamma+1/2}(x) \, dx, \qquad \gamma\ge3/2,
\end{equation}
which means that the semi-classical inequality holds true for all eigenvalues starting from $n=2$ and that in the latter inequality the Neumann boundary condition affects only the first eigenvalue.}

\medskip
\noindent
If $V\ge0$ is a diagonal $N\times N$ matrix-function, then the operator  $\mathcal H$ could be interpreted as a Schr\"odinger operator on a star graph with $N$ edges; the matrix $\mathfrak S$ describes a vertex coupling without the Dirichlet component \cite{Ku}. In such a case we obtain:

\begin{theorem}
Let $V\ge 0$ be a diagonal $N\times N$ matrix-function and let $\mathfrak S$ be a Hermitian matrix. Then the operator \eqref{eq:1} can be identified with a Schr\"odinger operator on a star graph with $N$ semi-infinite edges and its negative spectrum satisfies the inequality \eqref{LTh_main}.
\end{theorem}

%%%%%%

\medskip
\noindent
If both $V\ge0$ and $\mathfrak S$ are diagonal $N\times N$ matrices, then the negative spectrum of the operator $\mathcal H$ is the union of the eigenvalues from each channel and we obtain

%%%%%%%%%%%%%%

\begin{theorem}\label{split}
Let $V\ge 0$,  ${\rm Tr}\, V^2\in L^1(0,\infty)$ and let $V$ and $\mathfrak S$ be diagonal $N\times N$ matrices
with entries $v_{j}$ and $\sigma_{j}$, $j=1,\dots N$, respectively. Then the negative eigenvalues of the operator
$\mathcal H$ defined in \eqref{eq:1}, satisfy the inequality
\begin{equation}\label{channels}
\frac{3}{4} \sum_{j=1}^N \lambda_{j1}\, \sigma_j +   \frac12\, \sum_{j=1}^N \lambda_{j1}^{3/2}  + \sum_{j=1}^N\sum_{n=2}^\infty \lambda_{jn}^{3/2} \leq \frac{3}{16}\int_0^\infty {\rm Tr}\, V^2(x) \, dx+\frac{1}{4} \, {\rm Tr}\, {\mathfrak S^3},
\end{equation}
where $-\lambda_{jn}$ are negative eigenvalues of operators $h_j $ defined by
$$
h_j \psi(x) = \frac{d^2 }{dx^2} \psi(x) - v_j(x) \psi(x), \qquad \psi'(0) - \sigma_j \psi(0) = 0.
$$
\end{theorem}

\medskip
\noindent
{\bf Remark.} {\it Note that the inequality  \eqref{channels} is much more precise than \eqref{LTh_main} due to the diagonal structure of the operator $\mathcal H$. In \eqref{channels} all $N$ first eigenvalues generated by each channel are affected by the Robin boundary conditions, whereas in \eqref{LTh_main} only the first one, see Example 2b).}

\medskip
\noindent
Finally we give an example how our results could be applied for spectral estimates of multi-dimensional Schr\"odinger operators.

\noindent
Let $\Bbb R^d_+ = \{x=(x_1,x'): x_1>0, x'\in\Bbb R^{d-1}\}$ and let $H$ be a Schr\"odinger operator in $L^2(\Bbb R^d_+)$ with the Neumann boundary conditions
\begin{equation}\label{multiD}
H \psi = -\Delta \psi - V \psi = -\lambda \psi,   \qquad \frac{\partial}{\partial x_1} \psi(0,x') = 0.
\end{equation}
The following result could be obtained by a ``lifting" argument with respect to dimension, see \cite{L}, \cite{LT}:

%%%
\begin{theorem}\label{multi_dim}
Let $V\ge0$ and $V\in L^{\gamma+d/2}$, $\gamma \ge 3/2$. Then for the negative eigenvalues $\{-\lambda_n\}$ of the operator  \eqref{multiD} we have
\begin{multline}\label{LTh_d}
\sum_n \lambda_n^\gamma \le L_{\gamma,d}^{cl}\, \int_{\Bbb R^d_+} V^{\gamma+d/2}(x) \, dx +
\frac12\, L_{\gamma,d-1}^{cl} \, \int_{\Bbb R^{d-1}} \mu^{\gamma + (d-1)/2}_1(x')\, dx'\\
\le
 2 \, L_{\gamma,d}^{cl}\, \int_{\Bbb R^d_+} V^{\gamma+d/2} \, dx.
\end{multline}
 Here $\mu_1(x')$ is the ground state energy for the operator $ -d^2/dx_1^2 - V(x_1,x')$ in $L^2(\Bbb R_+)$ with the Neumann boundary condition at zero.
\end{theorem}

\medskip
\noindent
{\bf Remark.}
{\it A similar inequality could be obtained by extending the operator \eqref{multiD} to the whole space
$L^2(\Bbb R^d)$ with the symmetrically reflected potential. However, applying then the known Lieb-Thirring inequalities, we would have the constant $\,2^{\gamma+ d/2}$ instead of $\,2$ in \eqref{LTh_d}}.

%%%%%%%%%%%%%%%%%%%%%%%%%%%%%%%%%%%%%%%%%%%%%%%%%%

\section{Some auxiliary results}

\noindent
In this Section we assume that the matrix-function $V$ is compactly supported, ${\rm supp}\, V\subset [a,b]$ for some $a,b: 0<a<b<\infty$ and adapt the arguments from \cite{BL} to the case of semiaxis.

\noindent
We begin with stating a well-known fact concerning the ground state of the operator \eqref{eq:1}.
%%%%
\begin{lemma}\label{GrState}
Let $-\lambda<0$ be the ground state energy of the operator $\mathcal H$ and let
$\varphi(x) = \{\varphi_k\}_{k=0}^N$ be a $L^2(\Bbb R_+, \Bbb C^N)$-vector-function satisfying the equation
\begin{equation}\label{Schr2}
\mathcal H \varphi(x) = -\frac{d^2}{dx^2}\varphi(x) - V(x) \varphi(x) = -\lambda\, \varphi(x), \qquad
\varphi'(0) - \mathfrak S \varphi(0) = 0,
\end{equation}
and such that the $2N$ vector $(\varphi(0), \varphi'(0))$ is not trivial.
Then $\varphi(x) \not=0$, $x\in \Bbb R_+$, and the ground state energy multiplicity  is at most $N$.
\end{lemma}
%%%
\begin{proof} Suppose that $\varphi(x_0) = 0$ for some $x_0>0$. Consider the continuous function
\begin{equation*}
\tilde\varphi(x) =
\begin{cases}
\varphi(x), & x<x_0\\
0, & x\ge x_0.
\end{cases}
\end{equation*}
This function is non-trivial, belongs to the Sobolev space $H^1(\Bbb R_+,\Bbb C^N)$ and satisfies
\begin{multline*}
\int_{\Bbb R_+} \big(|\tilde\varphi' |^2 - (V\tilde\varphi,\tilde\varphi)_{\Bbb C^N} \big)\, dx =
\int_0^{x_0} \big(|\varphi' |^2 - (V\varphi,\varphi)_{\Bbb C^N} \big)\, dx\\
= \int_0^{x_0} (-\varphi'' - V\varphi, \varphi)_{\Bbb C^N}\, dx = -\lambda\,
\int_0^{x_0} |\varphi|^2\, dx = -\lambda \, \int_{\Bbb R_+} |\tilde\varphi|^2\, dx.
\end{multline*}
Therefore $\tilde\varphi$ minimizes the closed quadratic form associated with $\mathcal H$. Thus by the variational principle $\tilde\varphi$ belongs to the domain of $\mathcal H$ and solves
the Cauchy problem pointwise. However, since $\tilde\varphi(x) = 0 $ for $x\ge x_0$ it also
solves  the backward Cauchy problem with zero initial data at $x_0$ and by uniqueness must vanish
everywhere. This contradicts the non-triviality of $\tilde\varphi$ for $x<x_0$.
\end{proof}

\medskip
\noindent
Similarly to \cite{BL} let us introduce a (not necessary $L^2$) fundamental $N\times N$-matrix-solution $M(x)$ of the equation \eqref{Schr2}, where
$-\lambda$ is the ground state energy for the operator $\mathcal H$, so $M$ satisfies the equation
\begin{equation}\label{M_eq}
-\frac{d^2}{dx^2}M(x) - V(x) M(x) = -\lambda\, M(x), \qquad
M'(0) - \mathfrak S M(0) = 0.
\end{equation}
Denoting $M(0)=A$ and $M'(0) = B$, $B - \mathfrak S  \, A= 0$, we shall always assume that the matrix $A$ is invertible.

\noindent
By using Lemma \ref{GrState} we obtain that the matrix-function $M(x)$ is invertible for any $x\in\Bbb R_+$ and thus we can consider
\begin{equation}\label{F}
F(x) =  M'(x) \, M^{-1}(x).
\end{equation}
%

%%%%%%%%
\begin{lemma}\label{F_lemma}
The matrix function $F(x)$ satisfies the following properties:
\begin{itemize}
\item  $F(x)$ is Hermitian for any $x\in\Bbb R_+$.
\item $F(x)$ is independent of the choice of the matrices $A,B$, satisfying the equation $B-\mathfrak S A= 0$ and
$$
F(0) = B\,  A^{-1}  = \mathfrak S.
$$
\item $F$  satisfies the matrix Riccati equation
\begin{equation}\label{Ricc}
F'(x) + F^2(x) + V(x) = \lambda \mathbb{I}.
\end{equation}
\end{itemize}
\end{lemma}

\bigskip
\noindent
\begin{proof}
From the Wronskian identity
$$
\frac{d}{dx} W(x) := \frac{d}{dx} \Big( M^*(x)\, M'(x) \,  - (M^*(x))'\, M(x)\Big) = 0
$$
we obtain
$$
W(x) = M^*(x) \, M'(x)  - (M^*(x))'\, M(x) = {\rm const}.
$$
Since $M(0) = A$ and $M'(0) = B$, using the fact that $\mathfrak S$ is Hermitian we find
\begin{multline*}
W(0) = M^*(0) \, M'(0)  - (M^*(0))'\, M(0) \\
=  A^*\Big( B\, A^{-1} - (A^*)^{-1} B^*\Big)\, A = A^*(\mathfrak S-\mathfrak S^*) A = 0.
\end{multline*}
Thus
$$
W(x) = M^*(x) \, M'(x)  - (M^*(x))'\, M(x) =0.
$$
Multiplying the latter identity by $M^{-1}$ from the right and by  $(M^{-1})^*$ from the left we obtain
$F(x) = F^*(x)$.
Moreover
\begin{multline*}
F'+ F^2 = (M'\, M^{-1})' + (M'\, M^{-1})^2 \\
= M'' M^{-1} - M' \, M^{-1} M' M^{-1} + M' \, M^{-1} M' M^{-1} =
(\lambda - V) M\, M^{-1}  = \lambda \, \mathbb{I} - V.
\end{multline*}
\end{proof}
%%%%%%%%

\noindent
Next, we analyze the behavior of the matrices $F(x)$ and their eigenvalues and eigenvectors
as $x\to\infty$ . For $x > b$ any solution of the differential equation \eqref{M_eq} can be
written as
\begin{multline}\label{M_sol}
M(x)=\cosh (\sqrt \lambda(x-b))M(b)+ \frac{1}{\sqrt \lambda}\sinh (\sqrt \lambda (x-b))M'(b)\\
=\left(\cosh(\sqrt{\lambda}(x-b))\mathbb{I}+\frac{1}{\sqrt{\lambda}}\sinh(\sqrt{\lambda}(x-b))F(b)\right )M(b).
\end{multline}
With the help of this representation we show

\bigskip

%%%%%%%%
\begin{lemma}\label{F_spect}
For all $x\ge b$ it holds $F(x) = f(x,F(b))$, where
\begin{equation}\label{f}
f(x,\mu)=\sqrt\lambda \, \frac{\sqrt \lambda \tanh (\sqrt \lambda (x-b))
+\mu}{\sqrt \lambda +\mu \tanh(\sqrt \lambda(x-b))}.
\end{equation}
\end{lemma}
\begin{proof} In view of \eqref{M_sol} we have
$$
M'(x) = \left(\sqrt\lambda \sinh(\sqrt{\lambda}(x-b))\mathbb{I}+\cosh(\sqrt{\lambda}(x-b))F(b)\right )M(b),
$$
$$
(M(x))^{-1} = (M(b))^{-1} \left(\cosh(\sqrt{\lambda}(x-b))\mathbb{I}+\frac{1}{\sqrt{\lambda}}\sinh(\sqrt{\lambda}(x-b))F(b)\right )^{-1}.
$$
It remains to insert these expressions in the definition $F(x) = M'(x)(M(x))^{-1}$ and to
apply the spectral theorem for the Hermitian matrix $F(b)$.
\end{proof}

\medskip
\noindent
Note that $f(x, \mu)$ is strictly monotone in $\mu$. As a direct consequence of Lemma \ref{F_spect} we conclude, that the eigenvectors of the matrix $F(x)$ are independent of $x$ for $x\ge b$ as vectors in 
$\Bbb C^N$.  Moreover, the eigenvalues of $F$  may or may not depend on $x$ outside the support of $V$ depending on if they correspond to growing or decaying solutions.

%%%%%%%%%%%%
\begin{corollary}
Each eigenvalue $\mu_k$ of $F(b)$ gives rise to a continuous eigenvalue branch $\mu_k(x) = f(x,\mu_k(b))$. In particular, we have
$$
\mu_k(x) = -\sqrt\lambda \quad {\it iff} \quad \mu_k(b) = -\sqrt\lambda,
$$
and
$$
\lim_{x\to\infty} \mu_k(x) = \sqrt\lambda    \quad {\it iff} \quad \mu_k(b) \not= -\sqrt\lambda.
$$
The limit in the last expression is achieved exponentially fast.
\end{corollary}

%%%%%%%%%

\noindent
{\bf Remark.}
{\it There is a one-to-one correspondence between the $\varkappa_1$-dimensional
space of ground states for $\mathcal H$ and a $\varkappa_1$-dimensional eigenspace of $F(b)$ corresponding
to the eigenvalue $-\sqrt\lambda$. Indeed, since $M(x)$ is a fundamental system of the solutions
of the Cauchy problem  \eqref{Schr2} and $F(b)$ is invertible, any particular solution $\varphi$ of  \eqref{Schr2} can be represented as $\varphi(x) = F(x)(F(b))^{-1}\, \nu$ with some $\nu\in \Bbb C^N$.
Hence, by \eqref{M_sol}
\begin{multline}\label{part_sol}
\varphi_\nu(x)=
\cosh(\sqrt{\lambda}(x-b))\,\nu+\frac{1}{\sqrt{\lambda}}\sinh(\sqrt{\lambda}(x-b))F(b)\,\nu
\\
=\frac{1}{2\sqrt\lambda} \, e^{\sqrt\lambda(x-b)}\Big(\sqrt\lambda \nu + F(b) \, \nu\Big) -
\frac{1}{2\sqrt\lambda} \,  e^{-\sqrt\lambda(x-b)}\Big(\sqrt\lambda \nu - F(b) \, \nu\Big) .
\end{multline}
This function becomes an $L^2$-eigenfunction of $\mathcal H$, if and only if $F(b)\, \nu = -\sqrt\lambda \,\nu$.
}

%%%%%%%%%%%%%%%%%%%%%%%%%%%%%%%%%%%%%%%%%%%%%%%%%
%%%%%%%%%%%%%%%%%%%%%%%%%%%%%%%%%%%%%%%%%%%%%%%%%
%%%%%%%%%%%%%%%%%%%%%%%%%%%%%%%%%%%%%%%%%%%%%%%%%
\section{Proofs of the main results}

\noindent
{\it Proof of Theorem \ref{main}} .

\noindent
Let now $-\lambda_1$ be the ground state energy of the operator $\mathcal H$ with multiplicity $\varkappa_1 \le N$, let $M_1(x)$ be a fundamental system of solutions corresponding the eigenvalue $-\lambda_1$ and $F_1 = {M_1}'\, M_1^{-1}$.
We consider the operator
$$
Q_1=\frac{d}{dx}\otimes\mathbb{I} - F_1(x)
$$
and its adjoint
$$
Q_1^*=-\frac{d}{dx}\otimes\mathbb{I}-F_1(x)
$$
in $L^2(\mathbb{R}^+,\mathbb{C}^N)$. Using Riccati's equation \eqref{Ricc} we obtain the following factorization of the original operator $\mathcal H$
$$
Q_1^*Q_1=-\frac{d^2}{dx^2}\otimes\mathbb{I}+{F_1}'(x)+(F_1(x))^2=\mathcal H+\lambda_1\mathbb{I}.
$$
Consider
$$
Q_1 Q_1^*=-\frac{d^2}{dx^2}\otimes\mathbb{I}-V(x)-2{F_1}'(x)+\lambda_1\mathbb{I}=\mathcal H-2{F_1}'(x)+\lambda_1\mathbb{I}.
$$
Note that non-zero eigenvalues of $Q_1^*Q_1$ and $Q_1 Q_1^*$ are same.
However, while the vector-eigenfunctions $\varphi$ defined in \eqref{part_sol} satisfy the boundary conditions
$$
\varphi'(0)-\mathfrak S \varphi(0)=0,
$$
the vector-eigenfunctions of $Q_1 Q_1^*$ satisfy the Dirichlet boundary condition at $0$.

\medskip
\noindent
Indeed, if $\varphi$ is a vector-eigenfunction of $Q_1^* Q_1$ satisfying $\varphi'(0)- \mathfrak S\varphi(0)=0$ then
$\psi=Q_1\varphi$ is an eigenfunction of $Q_1 Q_1^*$
and
$$
\psi(0)=(Q_1\varphi)(0)=\varphi'(0)-F_1(0)\varphi(0)=0.
$$
Next, let us verify that the kernel ${\rm ker}\, Q_1^*$ is trivial, and consequently, $0\not\in \mathrm{spec}(Q_1Q_1^*)$.
Indeed, assume for a moment that there is a non-trivial vector-function $\psi$ satisfying the Dirichlet boundary conditions at $x=0$ and such that
\begin{equation}\label{QQ*}
Q_1Q_1^* \psi = 0.
\end{equation}
Then
$$
(Q_1Q_1^* \psi, \psi) = \|Q_1^*\psi\| = 0.
$$
However, $Q_1^*\phi=0$ if and only if $\psi'(x) = F(x)\psi(x)$ for all $x\in \Bbb R_+$ and, in particular,
$\psi'(0) = F(0) \psi(0) = 0$. Since $\psi$ satisfies the equation \eqref{QQ*} together with
$\psi(0)= \psi'(0) = 0$ we obtain that $\psi\equiv 0$.

\medskip
\noindent
Hence, the negative spectra of $\mathcal H$ and $\mathcal H- 2{F_1}'$ coincide except for the spectral
value of the ground state energy, which does not belong to the spectrum of $\mathcal H- 2{F_1}'$ anymore.
We emphasize that even in the case of a $\varkappa_1$-fold degenerate ground state $-\lambda_1 = -\lambda_2=\dots = -\lambda_{\varkappa_1}$ of $\mathcal H$, this commutation method removes all these eigenvalues
$-\lambda_1, -\lambda_2, \dots ,-\lambda_{\varkappa_1}$.

\medskip
\noindent
Therefore the spectral problem for the operator \eqref{eq:1} is reduced to the operator in $L^2(\Bbb R_+)$
\begin{equation*}
\mathcal H_1 \psi= \Big(-\frac{d^2}{dx^2}\otimes I - V(x) - 2{F_1}'\Big) \, \psi = -\lambda \psi \qquad \psi(0)=0.
\end{equation*}
Let us extend $V$ by zero to the negative semi-axis. Using then the variational principle we can apply the well-known Lieb-Thirring inequalities for 1D Schr\"odinger operators with matrix-valued potentials (see \cite{LW1}, \cite{BL}) and obtain
\begin{multline*}
\sum_{n=2}^\infty \varkappa_n\,\lambda_n^{3/2} \le \frac{3}{16} \, \int_0^\infty {\rm Tr}\, (V(x) + 2 {F_1}'(x)) ^2\, dx
\\
= \frac{3}{16} \, \int_0^\infty {\rm Tr}\,\Big(V^2(x) +4 F_1'(x) \big(V(x)  + F_1' (x)\big)\Big) \, dx.
\end{multline*}
Using the Riccati equation \eqref{Ricc}, the fact that  the matrix $\lim_{x\to\infty} F(x)$
has the eigenvalue $-\sqrt{\lambda_1}$ of multiplicity $\varkappa_1$ and the eigenvalue $\sqrt{\lambda_1}$ of multiplicity $N-\varkappa_1$ and that  $F(0) = \mathfrak S$, we finally arrive at
\begin{multline*}
\sum_{n=2}^\infty \varkappa_n\,\lambda_n^{3/2} \le \frac{3}{16} \, \int_0^\infty  {\rm Tr}\,\Big(V^2(x) +
4 F_1'(x) ( \lambda_1 - F_1^2(x) \Big)\, dx\\
=  \frac{3}{16} \, \int_0^\infty  {\rm Tr}\,V^2(x)\, dx + \frac{3}{4} \, \lambda_1 {\rm Tr}\, F_1(x)\Big|_0^\infty -
\frac14\, {\rm Tr}\, F_1^3(x)\Big|_0^\infty \\
= \frac{3}{16} \, \int_0^\infty  {\rm Tr}\,V^2(x)\, dx
+ \frac{3}{4} \, \lambda_1\Big(-\varkappa_1\sqrt{\lambda_1} +
(N-\varkappa_1) \sqrt{\lambda_1} -  {\rm Tr}\, \mathfrak S \Big)\\
-
\frac14\, \Big(-\varkappa_1\lambda_1^{3/2} +  (N-\varkappa_1)\lambda_1^{3/2} - {\rm Tr}\, \mathfrak S^{3}\Big)\\
=
\frac{3}{16} \, \int_0^\infty  {\rm Tr}\,V^2(x)\, dx - \frac12\, (2\varkappa_1 -N)\lambda_1^{3/2}  - \frac34 \, \lambda_1 \, {\rm Tr}\, \mathfrak S + \frac14\, {\rm Tr}\, \mathfrak S^{3}.
\end{multline*}
Finally using standard arguments we can consider the closure of the latter inequality from the class of compactly  supported potentials to the class $L^2(\Bbb R_+, \Bbb C^N \times \Bbb C^N)$.

\noindent
The proof of Theorem \ref{main} is complete.

%%%%%%%%%%%%%%%%%%%%%%%%%%%%%%%%%%%%%%%%%%%%%%
\medskip
\noindent
{\it Proof of Corollary \ref{AzL}}.

\noindent
Let us denote by $\lambda_n=\lambda_n(V)$ the eigenvalues of the Schr\"odinger
operator with the potential $V$.
Then by using the variational principle and the inequality \eqref{S<0} we find that for any $\gamma>3/2$
\begin{multline*}
{\mathcal B(\gamma-3/2,2)}\, \frac{3}{4} \, {\rm Tr}\, {\mathfrak S}\, \lambda_1^{\gamma-1/2}(V)  \\
+
{\mathcal B(\gamma-3/2,5/2)}\, \Big(\frac12\, (2\varkappa_1 -N) \, \lambda_1^{\gamma}(V)  + \sum_{n=2}^\infty \varkappa_n \lambda_n^{\gamma}(V)\Big) \\
=
\int_0^\infty \Big( \frac{3}{4}\, {\rm Tr}\, {\mathfrak S}\, (\lambda_1(V)-t)_+ \\
+ \frac12\, (2\varkappa_1 -N)
(\lambda_1(V)-t)_+^{3/2} + \sum_{n=2}^\infty \varkappa_n (\lambda_n(V) -t)_+^{3/2}\Big)
t^{\gamma-5/2}\, dt\\
\le
\int_0^\infty \Big( \frac{3}{4}\, {\rm Tr}\, {\mathfrak S}\, (\lambda_1((V-t)_+) \\
+ \frac12\, (2\varkappa_1 -N)
(\lambda_1((V-t)_+)^{3/2} + \sum_{n=2}^\infty \varkappa_n (\lambda_n((V -t)_+)^{3/2}\Big)
t^{\gamma-5/2}\, dt\\
\le \frac{3}{16} \, \int_0^\infty \int_0^\infty {\rm Tr}\, (V(x)-t)_+^2 \, t^{\gamma-5/2}\, dtdx\\
=
\mathcal B(\gamma-3/2,3)\,\frac{3}{16} \,  \int_0^\infty {\rm Tr}\, V^{\gamma + 1/2}(x)\, dx.
\end{multline*}
Dividing by ${\mathcal B(\gamma-3/2,5/2)}$ and noting that
$$
\frac{3}{16} \, \frac{\mathcal B(\gamma-3/2,3)}{\mathcal B(\gamma-3/2,5/2)} = L_{\gamma+1/2,1}^{cl}
$$
we complete the proof.

%%%%%%%%%%%%%%%%%%%%%%%%%%%%%%%%%%%%%%%%%%%%%%

\bigskip
\noindent
{\it Proof of Theorem \ref{multi_dim}}.

\noindent
Let $\{\mu_j(x')\}$ be eigenvalues of the Neumann problem for the Schr\"odinger operator
$$
-\frac{d^2}{dx_1^2} \psi(x_1, x')-V(x_1,x')\, \psi(x_1,x')= -\mu(x') \psi(x_1,x')
$$
considering $x'$ as a parameter.

\noindent
For any $\gamma\ge3/2$ and $d\ge1$ let us apply the operator version of the Lieb-Thirring inequality (see \cite{LW1}) with respect to $\Bbb R^{d-1}$ and obtain
\begin{equation*}
\sum_n \lambda_n^\gamma \le L_{\gamma,d-1}^{cl} \, \int_{\Bbb R^{d-1}} \sum_j \mu_j^{\gamma+(d-1)/2}(x')\, dx'.
\end{equation*}
By using \eqref{N-scalar} we find
\begin{multline*}
\sum_j \mu_j^{\gamma+(d-1)/2}(x')\le \frac12 \,\mu_1^{\gamma+(d-1)/2}(x') + L_{\gamma + (d-1)/2,1}^{cl} \, \int_0^\infty V^{\gamma+d/2}(x_1,x') \, dx_1 \\
\le 2\, L_{\gamma + (d-1)/2,1}^{cl} \, \int_0^\infty V^{\gamma+d/2}(x_1,x') \, dx_1.
\end{multline*}
Noticing that
$$
L_{\gamma,d-1}^{cl} L_{\gamma + (d-1)/2,1}^{cl} = L_{\gamma,d}^{cl}
$$
we obtain the proof.

%%%%%%%%%%%%%%%%%%%%%%%%%%%%%%%%%%%%%%%%%%%%%%%%%%%%%

\subsection*{Acknowledgment}

P.E. is grateful for the hospitality extended to him in \emph{Institut Mittag-Leffler} where a part of the work was done. The research was supported in part by the \emph{Czech Science Foundation} under the project P203/11/0701. 

%%%%%%%%%%%%%%%%%%%%%%%%%%%%%%%%%%%%%%%%%%%%%%%%%%%%%

\end{document}